\documentclass[11pt,a4paper]{article}
\usepackage[a4paper, total={6in, 8in}]{geometry}
\usepackage[utf8x]{inputenc}
\usepackage{amsmath, amsthm}
\usepackage{wrapfig,graphicx,amssymb,textcomp,array,amsmath}
\usepackage{algpseudocode} 
\usepackage{enumerate}
\usepackage{enumitem}
\usepackage{multirow}
\usepackage{tabularx}
\usepackage{algorithm}
\usepackage{color}
\usepackage{todonotes}
\usepackage[titletoc,title]{appendix}

\newcommand{\Kn}[1]{K#1}
\newcommand{\dimension}[1]{\text{dim}(#1)}
\newcommand{\removed}[1]{}

\title{Packing Plane Spanning Trees into a Point Set
}

\author{Ahmad Biniaz\thanks{University of Waterloo, Canada. {\tt ahmad.biniaz@gmail.com}}
	\and Alfredo Garc\'{\i}a\thanks{Universidad de Zaragoza, Spain. {\tt olaverri@unizar.es}}}

\date{}
\newtheorem{lemma}{Lemma}
\newtheorem{corollary}{Corollary}

\newtheorem{theorem}{Theorem}
\newtheorem{observation}{Observation}
\newtheorem*{problem*}{Problem}

\newtheorem*{invariant*}{Invariant}

\begin{document}
	\maketitle
	\begin{abstract}
		Let $P$ be a set of $n$ points in the plane in general position. We show that at least $\lfloor n/3\rfloor$ plane spanning trees can be packed into the complete geometric graph on $P$. This improves the previous best known lower bound $\Omega\left(\sqrt{n}\right)$. Towards our proof of this lower bound we show that the center of a set of points, in the $d$-dimensional space in general position, is of dimension either $0$ or $d$.
	\end{abstract}

	\section{Introduction}
	 In the two-dimensional space, a {\em geometric graph} $G$ is a graph whose vertices are points in the plane and whose edges are straight-line segments connecting the points. A subgraph $S$ of $G$ is {\em plane} if no pair of its edges cross each other. Two subgraphs $S_1$ and $S_2$ of $G$ are {\em edge-disjoint} if they do not share any edge. 
	 
	 Let $P$ be a set of $n$ points in the plane.
	 The {\em complete geometric graph} $\Kn{(P)}$ is the geometric graph with vertex set $P$ that has a straight-line edge between every pair of points in $P$. 
	We say that a sequence $S_1,S_2,S_3,\dots$ of subgraphs of $\Kn{(P)}$ is {\em packed into} $\Kn{(P)}$, if the subgraphs in this sequence are pairwise edge-disjoint. In a packing problem, we ask for the largest number of subgraphs of a given type that can be packed into $\Kn{(P)}$. Among all subgraphs, plane spanning trees, plane Hamiltonian paths, and plane perfect matchings are of interest. Since $\Kn{(P)}$ has ${n(n-1)}/{2}$ edges, at most $\lfloor n/2\rfloor$ spanning trees, at most $\lfloor n/2\rfloor$ Hamiltonian paths, and at most $n-1$ perfect matchings can be packed into it. 
	
	A long-standing open question is to determine whether or not it is possible to pack $\lfloor n/2\rfloor$ plane spanning trees into $\Kn{(P)}$. If $P$ is in convex position, the answer in the affirmative follows from the result of Bernhart and Kanien~\cite{Bernhart1979}, and a characterization of such plane spanning trees is given by Bose et al.~\cite{Bose2006}. In CCCG 2014, Aichholzer et al.~\cite{Aichholzer2017} showed that if $P$ is in general position (no three points on a line), then $\Omega(\sqrt{n})$ plane spanning trees can be packed into $\Kn{(P)}$; this bound is obtained by a clever combination of crossing family (a set of pairwise crossing edges)~\cite{Aronov1994} and double-stars (trees with only two interior nodes)~\cite{Bose2006}. Schnider~\cite{Schnider2016} showed that it is not always possible to pack $\lfloor n/2\rfloor$ plane spanning double stars into $\Kn{(P)}$, and gave a necessary and sufficient condition for the existence of such a packing. 
	As for packing other spanning structures into $\Kn{(P)}$, Aichholzer et al.~\cite{Aichholzer2017} and Biniaz et al.~\cite{Biniaz2015} showed a packing of 2 plane Hamiltonian cycles and a packing of $\lceil\log_2 n\rceil -2$ plane perfect matchings, respectively.
	
	The problem of packing spanning trees into (abstract) graphs is studied by Nash-Williams~\cite{Nash-Williams1961} and Tutte~\cite{Tutte1961} who independently obtained necessary and sufficient conditions to pack $k$ spanning trees into a graph. Kundu~\cite{Kundu1974} showed that at least $\lceil (k-1)/2\rceil$ spanning trees can be packed into any $k$-edge-connected graph.
	
	In this paper we show how to pack $\lfloor n/3\rfloor$ plane spanning trees into $\Kn{(P)}$ when $P$ is in general position. This improves the previous $\Omega(\sqrt{n})$ lower bound.
	
	\section{Packing Plane Spanning Trees}
	\label{proof-section}
	In this section we show how to pack $\lfloor {n}/{3}\rfloor$ plane spanning tree into $\Kn{(P)}$, where $P$ is a set of $n\geqslant 3$ points in the plane in general position (no three points on a line). 
	If $n\in\{3,4,5\}$ then one can easily find a plane spanning tree on $P$. Thus, we may assume that $n\geqslant 6$.
	
	The {\em center} of $P$ is a subset $C$ of the plane such that any closed
	halfplane intersecting $C$ contains at least $\lceil {n}/{3}\rceil$ points of $P$. A {\em centerpoint} of $P$ is a member of $C$, which does not necessarily belong to $P$. Thus, any halfplane that contains a centerpoint, has at least $\lceil {n}/{3}\rceil$ points of $P$. 
	It is well known that every point set in the plane has a centerpoint; see e.g. \cite[Chapter 4]{Edelsbrunner1987}. 
	We use the following corollary and lemma in our proof of the $\lfloor {n}/{3}\rfloor$ lower bound; the corollary follows from Theorem~\ref{alternate-dimension-thr} that we will prove later in Section~\ref{center-section}.
	
	\begin{corollary}
		\label{dimension-cor}
		Let $P$ be a set of $n\geqslant 6$ points in the plane in general position, and let $C$ be the center of $P$. Then, $C$ is either $2$-dimensional or $0$-dimensional. If $C$ is $0$-dimensional, then it consists of one point that belongs to $P$, moreover $n$ is of the form $3k+1$ for some integer $k\geqslant 2$.
	\end{corollary}

\begin{lemma}
	\label{prop1}
	Let $P$ be a set of $n$ points in the plane in general position, and let $c$ be a centerpoint of $P$. Then, for every point $p\in P$, each of the two closed halfplanes, that are determined by the line through $c$ and $p$, contains at least $\lceil n/3\rceil+1$ points of $P$.
\end{lemma}	
\begin{wrapfigure}{r}{1.1in} 
	\centering
	\vspace{-12pt} 
	\includegraphics[width=.95in]{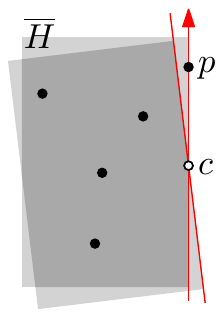} 
	\vspace{-5pt} 
\end{wrapfigure}
\noindent{{\em Proof.}}
For the sake of contradiction assume that a closed halfplane $\overline{H}$, that is determined by the line through $c$ and $p$, contains less than $\lceil n/3\rceil+1$ points of $P$. By symmetry assume that $\overline{H}$ is to the left side of this line oriented from $c$ to $p$; see the figure to the right. Since $c$ is a centerpoint and $\overline{H}$ contains $c$, the definition of centerpoint implies that $\overline{H}$ contains exactly $\lceil n/3\rceil$ points of $P$ (including $p$ and any other point of $P$ that may lie on the boundary of $\overline{H}$). By slightly rotating $\overline{H}$ counterclockwise around $c$, while keeping $c$ on the boundary of $\overline{H}$, we obtain a new closed halfplane that contains $c$ but misses $p$. This new halfplane contains less than $\lceil n/3\rceil$ points of $P$; this contradicts $c$ being a centerpoint of $P$. 
\qed \vspace{8pt}

	Now we proceed with our proof of the lower bound. We distinguish between two cases depending on whether the center $C$ of $P$ is $2$-dimensional or $0$-dimensional. First suppose that $C$ is $2$-dimensional. Then, $C$ contains a centerpoint, say $c$, that does not belong to $P$.
	Let $p_1,\dots,p_n$ be a counter-clockwise radial ordering of points in $P$ around $c$. For two points $p$ and $q$ in the plane, we denote by $\overrightarrow{pq}$, the ray emanating from $p$ that passes through $q$.

	Since every integer $n\geqslant 3$ has one of the forms $3k$, $3k+1$, and $3k+2$, for some $k\geqslant 1$, we will consider three cases. In each case, we show how to construct $k$ plane spanning directed graphs $G_1,\dots, G_k$ that are edge-disjoint. Then, for every $i\in\{1,\dots,k\}$, we obtain a plane spanning tree $T_i$ from $G_i$. First assume that $n=3k$. To build $G_i$, connect $p_i$ by outgoing edges to $p_{i+1},p_{i+2},\dots,p_{i+k}$, then connect $p_{i+k}$ by outgoing edges to $p_{i+k+1},p_{i+k+2},\dots,\allowbreak p_{i+2k}$, and then connect $p_{i+2k}$ by outgoing edges to $p_{i+2k+1},p_{i+2k+2},\dots,p_{i+3k}$, where all the indices are modulo $n$, and thus $p_{i+3k}=p_i$. The graph $G_i$, that is obtained this way, has one cycle $(p_i,p_{i+k},p_{i+2k},p_i)$; see Figure~\ref{graph-fig}. By Lemma~\ref{prop1}, every closed halfplane, that is determined by the line through $c$ and a point of $P$, contains at least $k+1$ points of $P$. Thus, all points $p_i,p_{i+1},\dots,p_{i+k}$ lie in the closed halfplane to the left of the line through $c$ and $p_i$ that is oriented from $c$ to $p_i$. Similarly, the points $p_{i+k},\dots,p_{i+2k}$ lie in the closed halfplane to the left of the oriented line from $c$ to $p_{i+k}$, and the points $p_{i+2k},\dots,p_{i+3k}$ lie in the closed halfplane to the left of the oriented line from $c$ to $p_{i+2k}$. Thus, all the $k$ edges outgoing from $p_i$ are in the convex wedge bounded by the rays $\overrightarrow{cp_i}$ and $\overrightarrow{cp_{i+k}}$, all the edges
	outgoing from $p_{i+k}$ are in the convex wedge bounded by $\overrightarrow{cp_{i+k}}$ and $\overrightarrow{c_{i+2k}}$, and all the edges from $p_{i+2k}$ are in the convex wedge bounded by $\overrightarrow{cp_{i+2k}}$ and $\overrightarrow{c_{i+3k}}$. Therefore, the spanning directed graph $G_i$ is plane. As depicted in Figure~\ref{graph-fig}, by removing the edge $(p_{i+2k},p_i)$ from $G_i$ we obtain a plane spanning (directed) tree $T_i$. This is the end of our construction of $k$ plane spanning trees.
	
		\begin{figure}[htb]
			\centering
			\setlength{\tabcolsep}{0in}
			$\begin{tabular}{cc}
			\multicolumn{1}{m{.5\columnwidth}}{\centering\includegraphics[width=.41\columnwidth]{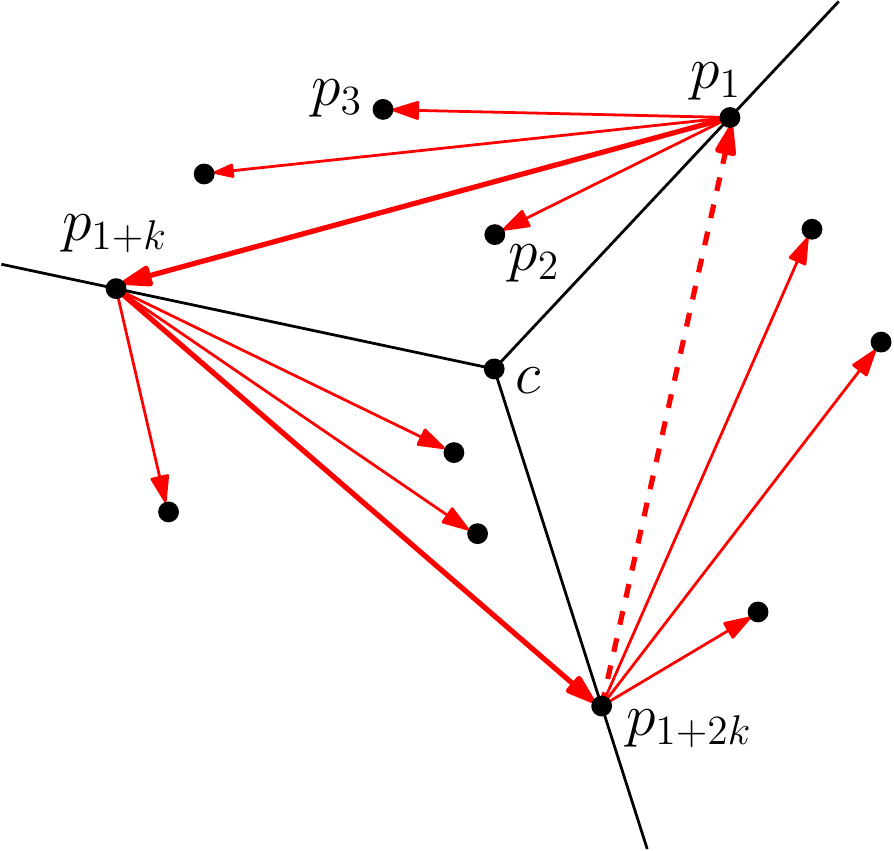}}
			&\multicolumn{1}{m{.5\columnwidth}}{\centering\includegraphics[width=.41\columnwidth]{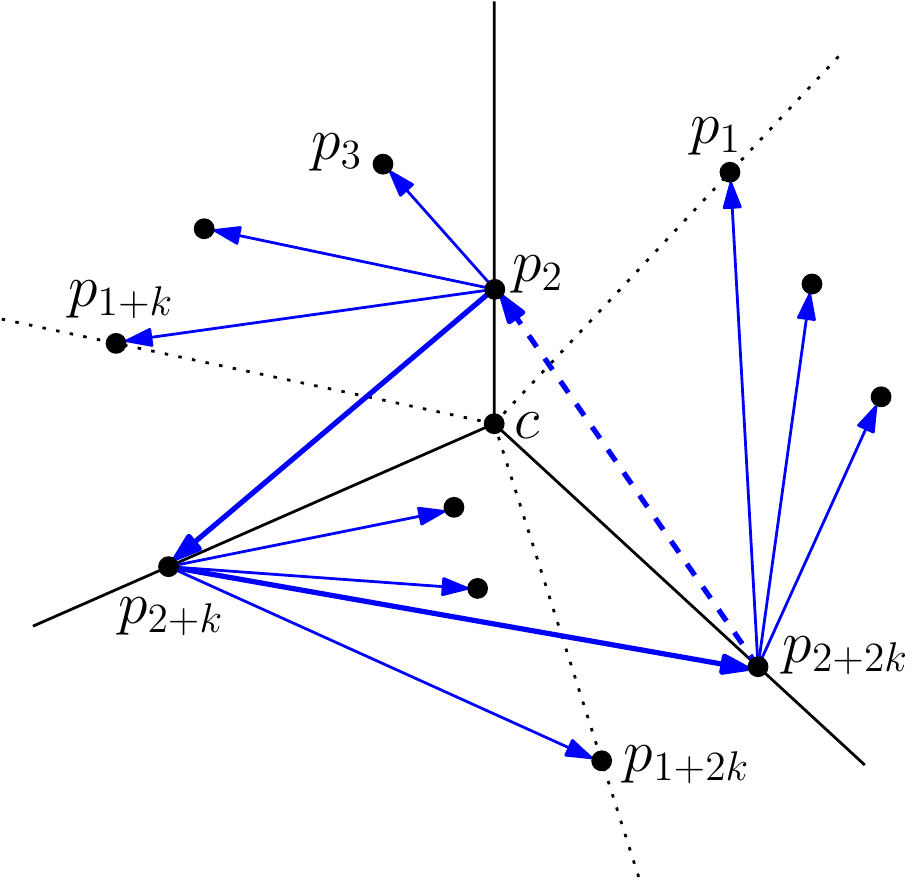}}
			\end{tabular}$
			\caption{The plane spanning trees $T_1$ (the left) and $T_2$ (the right) are obtained by removing the edges $(p_{1+2k},p_1)$ and $(p_{2+2k},p_2)$ from $G_1$ and $G_2$, respectively.}
			\label{graph-fig}
		\end{figure}

	To verify that the $k$ spanning trees obtained above are edge-disjoint, we show that two trees $T_i$ and $T_j$, with $i\neq j$, do not share any edge. Notice that the tail of every edge in $T_i$ belongs to the set $I=\{p_i,p_{i+k},p_{i+2k}\}$, and the tail of every edge in $T_j$ belongs to the set $J=\{p_j,p_{j+k},p_{j+2k}\}$, and $I\cap J=\emptyset$. For contrary, suppose that some edge $(p_r,p_s)$ belongs to both $T_i$ and $T_j$, and without loss of generality assume that in $T_i$ this edge is oriented from $p_r$ to $p_s$ while in $T_j$ it is oriented from $p_s$ to $p_r$. Then $p_r\in I$ and $p_s\in J$. Since $(p_r,p_s)\in T_i$ and the largest index of the head of every outgoing edge from $p_r$ is $r+k$, we have that $s\leqslant (r+k)\mod n$. Similarly, since $(p_s,p_r)\in T_j$ and the largest index of the head of every outgoing edge from $p_s$ is $s+k$, we have that $r\leqslant (s+k)\mod n$. However, these two inequalities cannot hold together; this contradicts our assumption that $(p_r,p_s)$ belongs to both trees. Thus, our claim, that $T_1,\dots, T_k$ are edge-disjoint, follows.
	This finishes our proof for the case where $n=3k$.
	
	If $n = 3k + 1$, then by Lemma~\ref{prop1}, every closed halfplane that is determined by the line through $c$ and a point of $P$ contains at least $k+2$ points of $P$. In this case, we construct $G_i$ by connecting $p_i$ to its following $k + 1$ points, i.e., $p_{i+1},\dots,p_{i+k+1}$, and then connecting each of $p_{i+k+1}$ and $p_{i+2k+1}$ to their following $k$ points. If $n = 3k + 2$, then we construct $G_i$ by connecting each of $p_i$ and $p_{i+k+1}$ to their following $k+1$ points, and then connecting $p_{i+2k+2}$ to its following $k$ points. This is the end of our proof for the case where $C$ is $2$-dimensional.
	
	Now we consider the case where $C$ is $0$-dimensional. By Corollary~\ref{dimension-cor}, $C$ consists of one point that belongs to $P$, and moreover $n=3k+1$ for some $k\geqslant 2$. Let $p\in P$ be the only point of $C$, and let $p_1,\dots,p_{n-1}$ be a counter-clockwise radial ordering of points in $P\setminus\{p\}$ around $p$. As in our first case (where $C$ was $2$-dimensional, $c$ was not in $P$, and $n$ was of the form $3k$) we construct $k$ edge-disjoint plane spanning trees $T_1,\dots, T_k$ on $P\setminus\{p\}$ where $p$ playing the role of $c$. Then, for every $i\in\{1,\dots,k\}$, by connecting $p$ to $p_i$, we obtain a plane spanning tree for $P$. These plane spanning trees are edge-disjoint. This is the end of our proof. In this section we have proved the following theorem.
	
	\begin{theorem}
		Every complete geometric graph, on a set of $n$ points in the plane in general position, contains at least $\lfloor n/3\rfloor$ edge-disjoint plane spanning trees.
	\end{theorem}
	
	\section{The Dimension of the Center of a Point Set}
	\label{center-section}
		
		The {\em center} of a set $P$ of $n\geqslant d+1$ points in $\mathbb{R}^d$ is a subset $C$ of $\mathbb{R}^d$ such that any closed halfspace intersecting $C$ contains at least $\alpha=\lceil {n}/{(d+1)}\rceil$ points of $P$. Based on this definition, one can characterize $C$ as the intersection of all closed halfspaces such that their complementary open halfspaces contain less than $\alpha$ points of $P$. 
		More precisely (see \cite[Chapter 4]{Edelsbrunner1987}) $C$ is the intersection of a finite set of closed halfspaces $\overline{H_1}, \overline{H_2},\dots, \overline{H_m}$ such that for each $\overline{H_i}$ 
		\begin{enumerate}
			\item the boundary of $\overline{H_i}$ contains at least $d$ affinely independent points of $P$, and
			\item the complementary open halfspace $H_i$ contains at most $\alpha - 1$ points of $P$, and the
			closure of $H_i$ contains at least $\alpha$ points of $P$.
		\end{enumerate}
		
		Being the intersection of closed halfspaces, $C$ is a convex polyhedron. A {\em centerpoint} of $P$ is a member of $C$, which does not necessarily belong to $P$. 
		It follows, from the definition of the center, that any halfspace containing a centerpoint has	at least $\alpha$ points of $P$. 
		It is well known that every point set in the plane has a centerpoint \cite[Chapter 4]{Edelsbrunner1987}. In dimensions 2 and 3, a centerpoint can be computed in $O(n)$ time \cite{Jadhav1994} and in $O(n^2)$ expected time \cite{Chan2004}, respectively.

		A set of points in $\mathbb{R}^d$, with $d\geqslant 2$, is said to be in {\em general position} if no $k+2$ of them lie in a $k$-dimensional flat for every $k\in \{1,\dots,d-1\}$.\footnote{A flat is a subset of $d$-dimensional space that is congruent to a Euclidean space of lower dimension. The flats in 2-dimensional space are points and lines, which have dimensions 0 and 1.} Alternatively, for a set of points in $\mathbb{R}^d$ to be in general position, it suffices that no $d+1$ of them lie on the same hyperplane.		
		In this section we prove that if a point set $P$ in $\mathbb{R}^d$ is in general position, then the center of $P$ is of dimension either 0 or $d$. Our proof of this claim uses the following result of Gr\"{u}nbaum.
		
		\begin{theorem}[Gr\"{u}nbaum, 1962 \cite{Grunbaum1962}]
			\label{Grunbaum-thr}
			Let $\mathcal{F}$ be a finite family of convex polyhedra in $\mathbb{R}^d$, let $I$ be their intersection, and let $s$ be an integer in $\{1,\dots, d\}$. If every intersection of $s$ members of $\mathcal{F}$ is of dimension $d$, but $I$ is $(d-s)$-dimensional, then there exist $s + 1$ members of $\mathcal{F}$ such that their intersection is $(d - s)$-dimensional.
		\end{theorem}
		
		\begin{figure}[htb]
			\centering
			\setlength{\tabcolsep}{0in}
			\vspace{10pt}
			$\begin{tabular}{ccc}
			\multicolumn{1}{m{.33\columnwidth}}{\centering\includegraphics[width=.26\columnwidth]{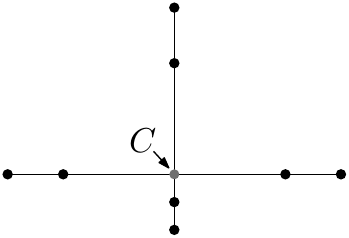}}
			&\multicolumn{1}{m{.33\columnwidth}}{\centering\includegraphics[width=.26\columnwidth]{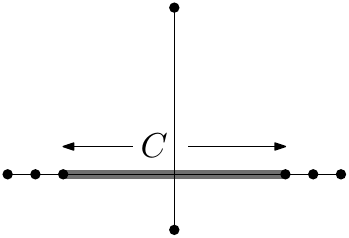}}
			&\multicolumn{1}{m{.33\columnwidth}}{\centering\includegraphics[width=.26\columnwidth]{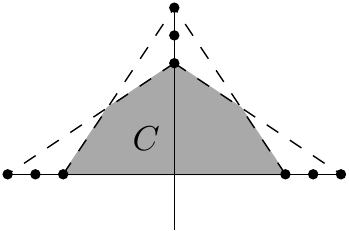}}
			\end{tabular}$
			\caption{The dimension of a point set in the plane, that is not in general position, can be any number in $\{0,1,2\}$.}
			\label{arbitrary-dim-fig}
			\vspace{5pt}
		\end{figure}
		
		Before proceeding to our proof, we note that if $P$ is not in general position, then the dimension of $C$ can be any number in $\{0,1,\dots,d\}$; see e.g. Figure~\ref{arbitrary-dim-fig} for the case where $d=2$.

\begin{observation}
	\label{dim-obs}
	For every $k\in\{1,\dots, d+1\}$ the dimension of a polyhedron defined by intersection of $k$ closed halfspaces in $\mathbb{R}^d$ is in the range $[d-k+1,d]$.
\end{observation}		
		
	\begin{theorem}
			\label{alternate-dimension-thr}
			Let $P$ be a set of $n\geqslant d+1$ points in $\mathbb{R}^d$, and let $C$ be the center of $P$. Then, $C$ is either $d$-dimensional, or contained in a $(d-s)$-dimensional polyhedron that has at least $n-(s+1)(\alpha -1)$ points of $P$ for some $s\in\{1,\dots,d\}$ and $\alpha=\lceil n/(d+1)\rceil$. In the latter case if $P$ is in general position and $n\geqslant d+3$, then $C$ consists of one point that belongs to $P$, and $n$ is of the form $k(d+1)+1$ for some integer $k\geqslant 2$.
		\end{theorem}
		
		\begin{proof}
			The center $C$ is a convex polyhedron that is the intersection of a finite family $\mathcal{H}$ of closed halfspaces such that each of their complementary	open halfspaces contains at most $\alpha-1$ points of $P$ \cite[Chapter 4]{Edelsbrunner1987}.
			Since $C$ is a convex polyhedron in $\mathbb{R}^d$, its dimension is in the range $[0,d]$.
			For the rest of the proof we consider the following two cases. 
			\begin{enumerate}
				\item[$($a$)$] The intersection of every $d+1$ members of $\mathcal{H}$ is of dimension $d$.
				\item[$($b$)$] The intersection of some $d+1$ members of $\mathcal{H}$ is of dimension less than $d$.
			\end{enumerate}
			
			First assume that we are in case (a). We prove that $C$ is $d$-dimensional. Our proof follows from Theorem~\ref{Grunbaum-thr} and a contrary argument. Assume that $C$ is not $d$-dimensional. Then, $C$ is $(d-s)$-dimensional for some $s\in\{1,\dots,d\}$. Since the intersection of every $s$ members of $\mathcal{H}$ is $d$-dimensional, by Theorem~\ref{Grunbaum-thr} there exist $s+1$ members of $\mathcal{H}$ whose intersection is $(d-s)$-dimensional. This contradicts the assumption of case (a) that the intersection of every $d+1$ members of $\mathcal{H}$ is $d$-dimensional. Therefore, $C$ is $d$-dimensional in this case.
			
			Now assume that we are in case (b). Let $s$ be the largest integer in $\{1,\dots, d\}$ such that every intersection of $s$ members of $\mathcal{H}$ is $d$-dimensional; notice that such an integer exists because every single halfspace in $\mathcal{H}$ is $d$-dimensional. Our choice of $s$ implies the existence of a subfamily $\mathcal{H}'$ of $s+1$ members of $\mathcal{H}$ whose intersection is $d'$-dimensional for some $d'<d$. Let $s'$ be an integer such that $d'=d-s'$. By Observation~\ref{dim-obs}, we have that $d'\geqslant d-s$, and equivalently $d-s'\geqslant d-s$; this implies $s'\leqslant s$. To this end we have a family $\mathcal{H}'$ with $s+1$ members for which every intersection of $s'$ members is $d$-dimensional (because $s'\leqslant s$ and $\mathcal{H}'\subseteq \mathcal{H}$), but the intersection of all members of $\mathcal{H'}$ is $(d-s')$-dimensional. Applying Theorem~\ref{Grunbaum-thr} on $\mathcal{H'}$ implies the existence of $s'+1$ members of $\mathcal{H'}$ whose intersection is $(d-s')$-dimensional. If $s'<s$, then this implies the existence of $s'+1\leqslant s$ members of $\mathcal{H}'\subseteq \mathcal{H}$, whose intersection is of dimension  $d-s'<d$. This contradicts the fact that the intersection of every $s$ members of $\mathcal{H}$ is $d$-dimensional. Thus, $s'=s$, and consequently, $d'=d-s'=d-s$. Therefore $C$ is contained in a $(d-s)$-dimensional polyhedron $I$ which is the intersection of the $s+1$ closed halfspaces of $\mathcal{H}'$.				
			Let $H_1,\ldots ,H_{s+1}$ be the complementary open halfspaces of members of $\mathcal{H}'$, and recall that each $H_i$ contains at most $\alpha-1$ points of $P$. Let $\overline{I}$ be the complement of $I$. Then,
			
			\begin{align}
			n&\notag =|I\cup \overline{I}|=|I\cup H_1\cup \dots\cup H_{s+1}|\\
			&\notag \leqslant |I|+|H_1|+\dots+|H_{s+1}| \leqslant |I|+ (s+1)(\alpha-1),
			\end{align}
			
			where we abuse the notations $I$, $\overline{I}$, and $H_i$ to refer to the subset of points of $P$ that they contain. This inequality implies that $I$ contains at least $n-(s+1)(\alpha-1)$ points of $P$. This finishes the proof of the theorem except for the part that $P$ is in general position.
			
			Now, assume that $P$ is in general position and $n\geqslant d+3$. By the definition of general position, the number of points of $P$ in a $(d-s)$-dimensional flat is not more than $d-s+1$. Since $I$ is $(d-s)$-dimensional, this implies that 
			 
			$$n-(s+1)(\alpha-1) \leqslant d-s+1.$$
			
			Notice that $n$ is of the form $k(d+1)+i$ for some integer $k\geqslant 1$ and some $i\in\{0,1,\dots, d\}$. Moreover, if $i$ is 0 or 1, then $k\geqslant 2$ because $n\geqslant d+3$. Now we consider two cases depending on whether or not $i$ is 0. If $i=0$, then $\alpha =k$. In this case, the above inequality simplifies to $k(d-s)\leqslant d-2s$, which is not possible because $k\geqslant 2$ and $d\geqslant s\geqslant 1$. If $i\in\{1,\dots,d\}$, then $\alpha=k+1$. In this case, the above inequality simplifies to $(k-1)(d-s)+i\leqslant 1$, which is not possible unless $d=s$ and $i=1$. Thus, for the above inequality to hold we should have $d=s$ and $i=1$. These two assertions imply that $n=k(d+1)+1$, and that $I$ is $0$-dimensional and consists of one point of $P$. Since $C\subseteq I$ and $C$ is not empty, $C$ also consists of one point of $P$.
\end{proof}

{\bf Acknowledgements.} Ahmad Biniaz was supported by NSERC Postdoctoral Fellowship. Alfredo Garc\'{\i}a was partially supported by MINECO project MTM2015-63791-R.

\vspace{5pt}

\begin{minipage}[l]{0.17\textwidth} \includegraphics[trim=10cm 6cm 10cm 5cm,clip,scale=0.15]{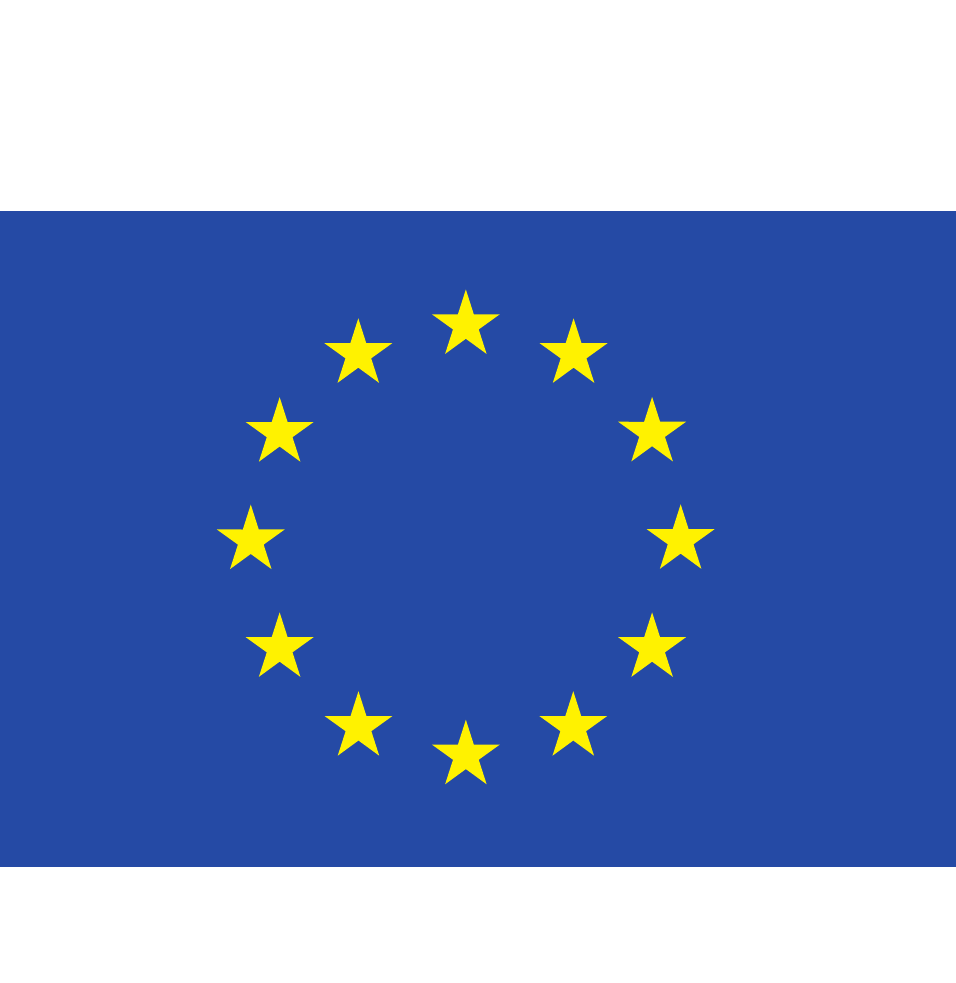} \end{minipage}  \hspace{-1cm} \begin{minipage}[l][1cm]{0.83\textwidth}
		Alfredo Garc\'{\i}a has received funding from the European Union's Horizon 2020 research and innovation programme under the Marie Sk\l{}odowska-Curie grant agreement No 734922.
\end{minipage}

	\bibliographystyle{abbrv}
	\bibliography{Packing-Spanning-Trees.bib}

\begin{thebibliography}{10}

\bibitem{Aichholzer2017}
O.~Aichholzer, T.~Hackl, M.~Korman, M.~J. van Kreveld, M.~L{\"{o}}ffler,
  A.~Pilz, B.~Speckmann, and E.~Welzl.
\newblock Packing plane spanning trees and paths in complete geometric graphs.
\newblock {\em Information Processing Letters}, 124:35--41, 2017.
\newblock Also in {CCCG'14}, pages 233--238.

\bibitem{Aronov1994}
B.~Aronov, P.~Erd{\"{o}}s, W.~Goddard, D.~J. Kleitman, M.~Klugerman, J.~Pach,
  and L.~J. Schulman.
\newblock Crossing families.
\newblock {\em Combinatorica}, 14(2):127--134, 1994.
\newblock Also in {SoCG'91}, pages 351--356.

\bibitem{Bernhart1979}
F.~Bernhart and P.~C. Kainen.
\newblock The book thickness of a graph.
\newblock {\em Journal of Combinatorial Theory, Series {B}}, 27(3):320--331,
  1979.

\bibitem{Biniaz2015}
A.~Biniaz, P.~Bose, A.~Maheshwari, and M.~H.~M. Smid.
\newblock Packing plane perfect matchings into a point set.
\newblock {\em Discrete Mathematics {\&} Theoretical Computer Science},
  17(2):119--142, 2015.

\bibitem{Bose2006}
P.~Bose, F.~Hurtado, E.~Rivera{-}Campo, and D.~R. Wood.
\newblock Partitions of complete geometric graphs into plane trees.
\newblock {\em Computational Geometry: Theory and Applications},
  34(2):116--125, 2006.

\bibitem{Chan2004}
T.~M. Chan.
\newblock An optimal randomized algorithm for maximum tukey depth.
\newblock In {\em Proceedings of the 15th Annual {ACM-SIAM} Symposium on
  Discrete Algorithms, {SODA}}, pages 430--436, 2004.

\bibitem{Edelsbrunner1987}
H.~Edelsbrunner.
\newblock {\em Algorithms in Combinatorial Geometry}.
\newblock Springer, 1987.

\bibitem{Grunbaum1962}
B.~Gr\"{u}nbaum.
\newblock The dimension of intersections of convex sets.
\newblock {\em Pacific Journal of Mathematics}, 12(1):197--202, 1962.

\bibitem{Jadhav1994}
S.~Jadhav and A.~Mukhopadhyay.
\newblock Computing a centerpoint of a finite planar set of points in linear
  time.
\newblock {\em Discrete {\&} Computational Geometry}, 12:291--312, 1994.

\bibitem{Kundu1974}
S.~Kundu.
\newblock Bounds on the number of disjoint spanning trees.
\newblock {\em Journal of Combinatorial Theory, Series B}, 17(2):199--203,
  1974.

\bibitem{Nash-Williams1961}
C.~{\relax St}. J.~A. Nash-Williams.
\newblock Edge-disjoint spanning trees of finite graphs.
\newblock {\em Journal of the London Mathematical Society}, 36(1):445--450,
  1961.

\bibitem{Schnider2016}
P.~Schnider.
\newblock Packing plane spanning double stars into complete geometric graphs.
\newblock In {\em Proceedings of the 32nd European Workshop on Computational
  Geometry, EuroCG}, pages 91--94, 2016.

\bibitem{Tutte1961}
W.~T. Tutte.
\newblock On the problem of decomposing a graph into $n$ connected factors.
\newblock {\em Journal of the London Mathematical Society}, 36(1):221--230,
  1961.

\end{thebibliography}
\end{document}